\documentclass[final,onefignum,onetabnum]{siamart171218}

\usepackage{lipsum}
\usepackage{amsfonts}
\usepackage{graphicx}
\usepackage{epstopdf}
\usepackage{todonotes}
\usepackage[ruled,vlined,algo2e]{algorithm2e}

\DeclareMathOperator*{\argmax}{argmax}
\DeclareMathOperator*{\argmin}{argmin}

\ifpdf
  \DeclareGraphicsExtensions{.eps,.pdf,.png,.jpg}
\else
  \DeclareGraphicsExtensions{.eps}
\fi

\newsiamremark{remark}{Remark}
\newsiamremark{hypothesis}{Hypothesis}
\crefname{hypothesis}{Hypothesis}{Hypotheses}
\newsiamthm{claim}{Claim}
\newtheorem*{theorem-non}{Theorem}

\headers{Adaptive Test Allocation for Outbreak Tracking}{P. Batlle, C. Fernandez-Granda, J. Bruna and V. M. Preciado}

\title{Adaptive Test Allocation for Outbreak Detection and Tracking in Social Contact Networks\thanks{\textbf{Funding: }Pau Batlle was partially supported by Capital One Inc. and the Cellex Foundation. Joan Bruna acknowledges partial support by Alfred P. Sloan Foundation, NSF RI-1816753 and NSF CAREER CIF 1845360. Carlos Fernandez-Granda was partially supported by NSF DMS 2009752. Victor M. Preciado was partially supported by NSF-CAREER-ECCS-1651433, NSF-III-200884556, and the Rockefeller Foundation.
}}
\author{Pau Batlle\thanks{Courant Institute of Mathematical Sciences, Center for Data Science, New York University, NY.}
\and Joan Bruna\footnotemark[2]
\and Carlos Fernandez-Granda\footnotemark[2] 
\and Victor M. Preciado\thanks{Department of Electrical and Systems Engineering, Applied Mathematics and Computational Science, University of Pennsylvania, PA}}

\usepackage{amsopn}

\makeatletter
\newcommand*{\addFileDependency}[1]{
  \typeout{(#1)}
  \@addtofilelist{#1}
  \IfFileExists{#1}{}{\typeout{No file #1.}}
}
\makeatother

\newcommand{\PP}{\mathbb{P}}
\newcommand\numberthis{\addtocounter{equation}{1}\tag{\theequation}}
\newtheorem{prop}{Proposition}

\begin{document}

\maketitle

\begin{abstract}
We present a general framework for adaptive allocation of viral tests in social contact networks. 
We pose and solve several complementary problems. First, we consider the design of a social sensing system whose objective is the early detection of a novel epidemic outbreak. In particular, we propose an algorithm to select a subset of individuals to be tested in order to detect the onset of an epidemic outbreak as fast as possible. We pose this problem as a hitting time probability maximization problem and use submodularity optimization techniques to derive explicit quality guarantees for the proposed solution.
Second, once an epidemic outbreak has been detected, we consider the problem of adaptively distributing viral tests over time in order to maximize the information gained about the current state of the epidemic. We formalize this problem in terms of information entropy and mutual information and propose an adaptive allocation strategy with quality guarantees.
For these problems, we derive analytical solutions for any stochastic compartmental epidemic model with Markovian dynamics, as well as efficient Monte-Carlo-based algorithms for non-Markovian dynamics.
Finally, we illustrate the performance of the proposed framework in numerical experiments involving a model of Covid-19 applied to a real human contact network.
\end{abstract}

\begin{keywords}
Epidemiology, Social networks, Markov chains, Submodular optimization, entropy-based sampling
\end{keywords}

\begin{AMS}
92D30, 91D30, 60J20, 90C35, 94A17
\end{AMS}

\section{Introduction}
\label{sec:intro}
In December $31^{st}$ 2019, The Municipal Health Commission of Wuhan
(China) reported a cluster of cases of pneumonia caused by a novel
coronavirus \cite{WHO_timeline}. This new virus rapidly propagated
worldwide through the air transportation network and many
countries decided to implement severe mobility restrictions and social
distancing policies to ``flatten the curve'' of the pandemic. However, as society reopens, mobility increases, and social distancing relaxes, new epidemic outbreaks become a very real threat. In this situation, it is of upmost societal importance to develop efficient strategies for early detection and tracking of epidemic outbreaks.

In this paper, we study the problem of allocating viral tests \cite{CDC_ViralTests}
 in order to (\emph{i}) detect a novel epidemic outbreak as early as possible, as well as (\emph{ii}) to retrieve as much information as possible about the evolution of the epidemic. In our work, we consider a social contact network over which a disease is spreading according to a stochastic compartmental model \cite{SureveyVP}. The main questions explored in this article are:
 
 \smallskip
 
1.\emph{ Early detection of epidemic outbreaks with limited viral tests}:
In particular, what nodes should we test in a contact social graph to maximize the probability of early detection? We will pose this problem in terms of hitting times of an stochastic process associated to the social graph and propose an algorithm to solve it with quality guarantees based on submodular optimization.

\smallskip

2.\emph{ Estimation of past and current state of the disease}: Given the results of a collection
of viral tests, what are the probabilities of infection for each individual
in the social network? Furthermore, we analyze the past evolution of the epidemic to estimate where and when the infection is most likely to have started.

\smallskip

3.\emph{ Adaptive test allocation for epidemic tracking}: Once an epidemic
outbreak has been detected, how should we dynamically allocate viral tests
to gain as much information as possible about the current state of the epidemic?

\smallskip

General work on stochastic compartmental models in networks include \cite{Meyers2006}, \cite{RePEc:wop:safiwp:01-12-073}, \cite{Kiss2017}. Additionally, \cite{SureveyVP} and \cite{Britton2019EpidemicMO} provide a general survey of problems involving spreading processes in networks. In \cite{VanMieghem2009}, van Mieghen et al. study the spread of malware in computer networks using Markov chains; however, their focus is on mean-field approximations derived from continuous-time Markov chains, while our work focuses on the exact stochastic model of the epidemic process. In \cite{Leskovec2007}, Leskovec et al. propose a sensor placement framework, similar to the test placement in our work, to detect outbreaks of water contaminants and other spreads; however, the authors use a deterministic propagation model over a directed network, instead of a stochastic epidemic model. In \cite{Shah2011}, Shah et al. study the culprit detection problem for the popular SI epidemic model on a network using the so-called rumor centrality. In \cite{Spinelli2019}, Spinelli et al. also study the culprit detection problem for a specific family of spreading models without any concerns about early detection. In  \cite{Yan2020}, Yan et al. consider the independent cascade model and study the problem of immunizing edges in order to minimize the expected number of infected nodes at the conclusion of the spreading process. As far as testing is concerned, there is literature that gives results on different network monitoring techniques \cite{Christakis2010} \cite{GarciaHerranz2014} \cite{Salathe2010} \cite{Sun2014} \cite{Herrera2016}; however, these works do not aim to find an optimal solution according to any metric, but analyze the performance of particular heuristics. Finally, the works in \cite{Borgatti2005}, \cite{Christley2005} and \cite{Kitsak2010} analyze several heuristics for epidemic detection based on different network centrality measures.

The article is organized as follows. In \cref{Theo}, we formalize our theoretical setup and discuss the models to which this framework is applicable. The three questions described above are explored in \cref{EarlyDet}, \cref{Inference} and \cref{DynamicAllocation}, respectively. Finally, \cref{Experiments} presents experiments in a real dataset of human interactions where we apply our framework using a realistic model of the spread of Covid-19 \cite{Ferretti2020}.

\section{Notation and preliminaries}
\label{Theo}
 For a given $n\in \mathbb N$, we let $[n]$ be the set $\{1,\ldots,n\}$. We consider a given network $G=(V,E)$ where $V=[n]$, and a continuous-time stochastic compartmental epidemic process, denoted by $\{X(t)\}_{t\geq 0}$, running over $G$. At every $t \in \mathbb{R}_{+}$, each of the $n$ nodes in the network is in one out of $s$ possible states, where each state represents a compartment in the epidemic model. Since we have $n$ nodes, the networked stochastic process $\{X(t)\}_{t\geq 0}$ has a finite state space $\mathcal{S}$ with $|\mathcal{S}| = s^n$. One of the simplest networked compartmental models is the SIR model \cite{1927}, which presents three compartments: Susceptible, Infectious, and Removed. In this model, infectious nodes may infect healthy neighbors with probability rate $\beta$ and may transition into the removed compartment (i.e., no longer infectious) with probability rate $\gamma$.
 
 In the rest of the paper, we assume that the initial state $X(0)$ of the epidemic process is randomly chosen from a known probability distribution $D$ supported in $\mathcal{S}$, and that all subsequent probabilities are conditioned on the realization of $X(0)$. If the epidemic process $\{X(t)\}_{t\geq 0}$ is Markovian, we can derive analytical solutions to the problems under consideration (shown in  \cref{Markovchains}). However, these analytical solutions are usable in practice only for relatively small graphs. In the following sections, we will provide computational tools to analyze non-Markovian epidemic processes running over large graphs. In this more general case, instead of using untractable analytical solutions, we will provide efficient numerical algorithms to solve different problems of interest.

\section{Early detection of epidemic outbreaks with limited viral tests}
\label{EarlyDet}
The first key question we address is how to optimally monitor a contact network for early detection of a new outbreak with limited resources. We assume that we are able to continuously monitor the health of $k$ nodes of the network before the onset of an outbreak. We aim to answer two optimization questions: 

\smallskip

\textbf{Q1A (Test placement with monitoring constraint):} \emph{For a given $k\in\mathbb{N}$ and $\tau>0$, which $k$ nodes should we continuously monitor to detect a novel outbreak before a certain time $\tau$ (counting from the onset of the outbreak) with the highest possible probability?}

\smallskip

\textbf{Q1B (Test placement with probability constraint):} \emph{Given a threshold time $\tau>0$ and a probability $P$, what is the minimum number $k$ of nodes we need to monitor to detect the epidemic outbreak before time $\tau$ with a probability $P$? Where should we place them?}

\smallskip

To analyze these questions, we assume that those nodes being monitored are frequently tested. We assume that the available tests provide partial information about the state of the node. In particular, we consider a partition of the set of $s$ possible states into two non-empty subsets, $G_+$ and $G_-$, and assume that the test is able to determine in what subset the state of the monitored node is. In practice, the set $G_+$ (resp., $G_-$) represent node states that would result in a positive (resp., negative) viral test result. We say that the a node is \textit{detectable} if its state is in $G_+$, and that the epidemic is detected when one of the monitored nodes becomes detectable for the first time.

In order to track the time it takes for the epidemic to be detected, we use the concept of stopping time of a stochastic process. Given a network stochastic process $\{X(t)\}_{t\geq0}$ and a subset $A\subset \mathcal{S}$, its \textit{stopping time} $T_A$ is defined as the random variable $\min\{t \geq 0 : X(t) \in A\}$, where $X(t) \in \mathcal{S}$ is the state of the stochastic process at time $t$. If the process never reaches $A$, we set $T_A = \infty$. Given a subset of nodes $W \subset V$, we additionally define the \textit{detection set} of $W$, denoted by $D_W$, as the subset of $\mathcal{S}$ consisting of those network states in which at least one of the nodes in $W$ is in a detectable state (i.e., one of the nodes in $W$ test positive). This means that if we monitor the nodes in $W$, the epidemic outbreak is detected when the network process $\{X(t)\}_{t\geq0}$ reaches one of the states in $D_W$; however, the exact network state will still be unknown.

Now, Question \textbf{Q1A} can be formalized as follows:
Given a time horizon $\tau >0$, we want to monitor $k$ nodes of the network in order to maximize the probability that the process reaches the detection set $D_W$ before time $t=\tau$ (counting from the onset of the epidemic outbreak). Hence, the optimal set of nodes to be monitored can be found as the solution of the following optimization problem:
\begin{align}
\argmax_{W\subset V, |W| = k} \mathbb{P}\left(T_{D_W} \leq \tau \right). \tag{Q1A}
\end{align}
Similarly, the answer to Question \textbf{Q1B} is the solution to the following optimization:
\begin{align}
\argmin_{W\subset V \text{ s.t } \mathbb{P}\left(T_{D_W} \leq \tau \right) \geq P} |W|, \tag{Q1B}
\end{align}
i.e, the smallest set of nodes that we need to monitor such that the probability of detecting the epidemic outbreak before time $t=\tau$ is greater than $P$. We conveniently define the optimization objective function over subsets of nodes for a given $\tau$ as $f_\tau \colon W \mapsto \mathbb{P}\left(T_{D_W} \leq \tau\right)$, so that equations (\textbf{Q1A}) and (\textbf{Q1B}) can be written, respectively, as
\begin{align}
\argmax_{W\subset V, |W| = k} f_\tau(W) \text{\quad and }~ \argmin_{W\subset V, f_\tau(W) \geq P} |W|~.
\end{align}

Notice that these are combinatorial optimization problems and finding the optimal solutions is exponentially hard. In the rest of the paper, we focus on finding approximate solutions with quality guarantees. In this direction, there are two separate subproblems we need to address: First of all, the function  $f_\tau$ can only be computed for Markovian epidemic processes taking place in small networks (see \cref{Proofs}); hence, we need to approximate this objective function for non-Markovian processes over large networks. Secondly, we also need an optimization scheme to find an approximated solution with qualities guarantees (without evaluating $f_\tau$ an exponential number of times).

\subsection{Function evaluation}
\label{f_hat_def}
The function $f_\tau$ can be approximated using Monte Carlo samples, as described below. First, we simulate the stochastic epidemic process $N_R$ times, where each simulation will be stopped when one of two things happen: Either we reach an absorbing state, or all the nodes have already reached a detectable state at least once. Then, $f_\tau$ can be approximated as follows: Let $L$ be a $N_R\times n$ matrix such that, for every run of the process $r \in [N_R]$
\begin{align}
L[r,j] = \min\{{T \in \mathbb{R_+} \colon X_j(T)\in G_+\text{ in run }r}\} ~,
\end{align}
where $X_j(t)$ is the state of node $j$ at time $t$ and $L[r,j] = \infty$ if node $j$ is never detectable in run $r$. Given the matrix $L$ and any time $\tau$, an estimator for $f_\tau$, denoted by $\hat f_\tau$, can be calculated as follows: $\hat{f}_\tau(W)=\frac{1}{N_R}|\{r \in [N_R] \colon \min_{i\in W}{L[r,i]} \leq \tau\}|$. As $N_R \rightarrow \infty$, we have that $\hat{f}_\tau \rightarrow f_\tau$ uniformly, because of the law of large numbers.

\subsection{Optimization of $f_\tau$ via submodularity}
The combinatorial structure of the problem requires not only a way to rapidly evaluate the objective function but an optimization scheme that avoids evaluating an exponential number of possible node monitorizations. In order to do that, we prove and exploit the submodularity properties of $f_\tau$ combined with fundamental results about submodular optimization. 

If $\Omega$ is a finite set, a function $h \colon \mathcal{P}(\Omega) \rightarrow \mathbb{R}$ is called \emph{submodular} if it satisfies one of these three equivalent conditions: (\textbf{Condition }1) $\forall X, Y \subseteq \Omega$ with $X \subseteq Y$ and every ${\displaystyle x\in \Omega \setminus Y}$, we have that ${\displaystyle h(X\cup \{x\})-h(X)\geq h(Y\cup \{x\})-h(Y)}$; \\ (\textbf{Condition }2) $\forall S, T \subseteq \Omega$ we have that $ h(S)+h(T)\geq h(S\cup T)+h(S\cap T)$; \\ (\textbf{Condition }3) $\forall X\subseteq \Omega$ and ${\displaystyle x_{1},x_{2}\in \Omega \backslash X}$ such that ${\displaystyle x_{1}\neq x_{2}}$, $h(X\cup \{x_{1}\})+h(X\cup \{x_{2}\})\geq h(X\cup \{x_{1},x_{2}\})+h(X)$. We aim to prove that $f_\tau$ is  a non-negative, monotone (i.e., $f_\tau(X) \leq f_\tau(Y)$ for $X\subset Y$) and submodular function. The non-negativity is trivial from the definition of probability, and monotonicity comes from the fact that for $A\subset B$, $D_A \subset D_B$, and so the event $T_{D_A} \leq \tau$ implies that $T_{D_B} \leq \tau$, hence, $f_\tau(A) \leq f_\tau(B)$. Furthermore, $f_\tau$ is submodular (as proved in  \cref{Proofs}).

\begin{theorem}
\label{thm:submod}
The set-function $f_\tau \colon W \mapsto \mathbb{P}\left(T_{D_W} \leq \tau\right)$ is submodular.
\end{theorem}

We can now invoke two well-known results in submodular optimization theory to derive quality guarantees of greedy-like optimization schemes aiming to solve Problems (\textbf{Q1A}) and (\textbf{Q1B}), using \cref{alg:greedy1} and \cref{alg:greedy2}, described below. Both algorithms run in polynomial time and only require $\mathcal{O}(n^2)$ evaluations of the objective function.

\begin{theorem}[\cite{Nemhauser1978}]
\label{thm:alg1}
If a set-function $f$ is monotone, submodular and non-negative, the greedy scheme in \cref{alg:greedy1} applied to Problem (\textbf{Q1A}) returns a solution $S'$ for which $f(S') \geq (1-\frac{1}{e})f(S^*)$ where $S^*$ is the optimal set.
\end{theorem}
\begin{algorithm}[H]
\textbf{Input:} $k \in \mathbb{N}, f$ function over subsets of V \\
\textbf{Output:} $S \subset V$ with $|S| = k$, an approximate solution to (Q1A) \\
$S \leftarrow \emptyset$, $i \leftarrow 0$ \\
\While{$i\leq k$}{
$S \leftarrow S \cup \displaystyle \argmax_{v \in V \setminus S}{f(S\cup \{v\})}$\\
$i \leftarrow i+1$}
\Return{S}
\caption{Greedy scheme applicable to (Q1A) when $f=f_\tau$}
\label{alg:greedy1}
\end{algorithm}
\begin{theorem}[\cite{Wolsey1982}]
\label{thm:alg2} 
If $f$ is monotone and submodular, the greedy scheme in \cref{alg:greedy2} applied to Problem (\textbf{Q1B}) returns a solution $S'$ for which~$\dfrac{|S'|}{|S^*|} \leq 1 + \log\dfrac{f(V)-f(\emptyset)}{f(S')-f(S_{-1})}$, where $S^*$ is the optimal set and $S_{-1}$ is the solution set at the iteration prior to the termination of \cref{alg:greedy2}.
\end{theorem}
\begin{algorithm}[H]
\SetKwData{Left}{left}\SetKwData{This}{this}\SetKwData{Up}{up}
\SetKwFunction{Union}{Union}\SetKwFunction{FindCompress}{FindCompress}
\textbf{Input:} $P \in [0,1], f$ set-function over subsets of V \\
\textbf{Output:} $S \subset V$ with $f(S) \geq P$, an approximate solution to (Q1B) \\
$S \leftarrow \emptyset$ \;\\
\While{$f(S) < P$}{
$S \leftarrow S \cup \displaystyle \argmax_{v \in V \setminus S}{f(S\cup \{v\})}$}
\Return{S}
\caption{Greedy scheme applicable to (Q1B) when $f=f_\tau$}
\label{alg:greedy2}
\end{algorithm}
Note that in the case of non-Markovian epidemic models and/or large networks, we cannot directly evaluate $f_\tau$, but an approximation $\hat{f}_\tau$. A natural question is whether $\hat f_\tau$ has similar properties as $f_\tau$, so that we can guarantee quality of the optimization.
\begin{theorem}
\label{thm:submodfhat}
The approximation function $\hat{f}_\tau$, defined in \cref{f_hat_def}, is non-negative, monotone and submodular for all $N_R \in \mathbb{N}$
\end{theorem}
Using this result (proved in \cref{Proofs}), we conclude that the quality guarantees in \cref{thm:alg1}
and \cref{thm:alg2} are also applicable to the approximation function $\hat f_\tau$. 

\subsection{Toy example in a small network}
\label{toy1}
We illustrate our procedures using the graph in \cref{fig:petit}. We use a SIR model with $\beta = 0.5$, $\delta = 0.25$, and a single initially infected node chosen uniformly at random. Setting $\tau = 0.5$, the set of $k = 2$ nodes to be monitored such that $f_\tau$ is maximized is $\{1,5\}$ (circled in black in the figure) with a value of $f_\tau(\{1,5\})=0.442$; in other words, monitoring these two nodes, we are able to detect the epidemic outbreak before $0.5$ time units with a probability equal to $0.442$. This solution is obtained via an exhaustive combinatorial search. If, in contrast, we use the greedy-scheme in \cref{alg:greedy1}, we obtain $\{3,0\}$ as our approximate solution and $f_\tau(\{3,0\}) = 0.438$. \cref{thm:submodfhat} ensures that the greedy solution (i.e., $0.438$) is not worse than $(1-1/e)\times 0.442=0.279$ (notice that the greedy solution is much better than that worst case value).
\begin{figure}[H]
    \centering
    \includegraphics[scale = 0.5]{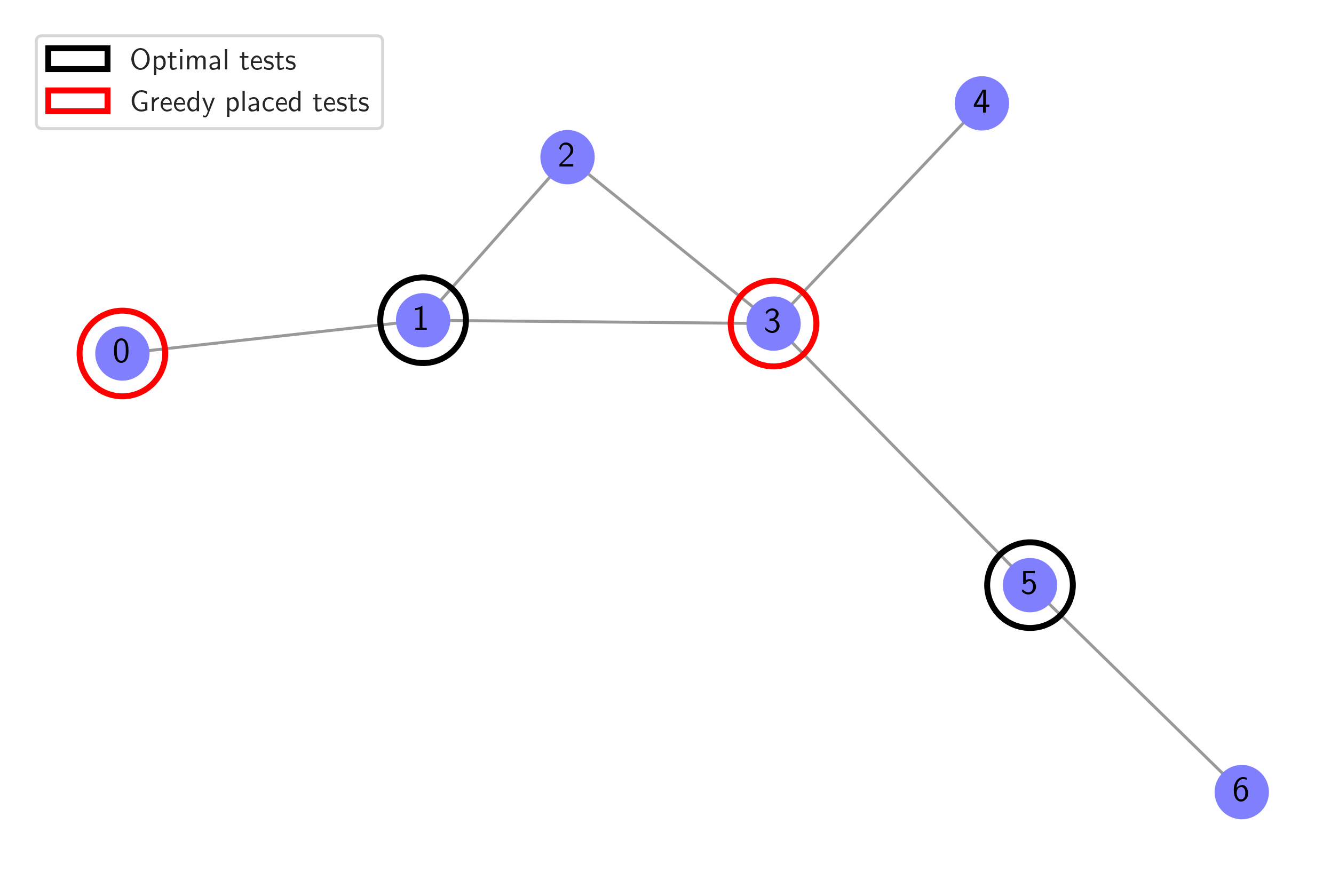}
    \caption{Toy example used in \cref{toy1}, with $n = 7$ nodes. For $\tau = 0.5$, nodes $1$ and $5$ are the optimal set, but a greedy approach selects nodes 3 and 0.}
    \label{fig:petit}
\end{figure}

\section{Estimation of past and current state of the disease}
\label{Inference}
Once an epidemic outbreak has been detected, it is of practical interest to use the results of the tests used during the monitoring phase to estimate the network state of the disease. In this section, estimate the global state of the network using the information obtained from the viral test results retrieved from a subset of nodes. In this direction, given test results for a subset of nodes, we formulate two different subquestions: 

\smallskip

\textbf{Q2A (Patient zero detection):} \emph{What is the probability of each node being patient zero?}

\smallskip

\textbf{Q2B (Outbreak time estimation):} \emph{How much time has passed since the outbreak started?}

\smallskip

\textbf{Q2C (Current network status assessment):} \emph{What is the probability of each individual node being infected?}

\smallskip

Assuming that nodes $v_1,\ldots,v_k\in V$ are our $k$ monitoring nodes, we define the $k$ dimensional vector $O$ such that $O_i=1$ when $v_i$ have tested positive during the monitoring phase and $O_i=0$ otherwise. Hence, assuming that $x_i\in \mathcal{S}$ is the network state in which only node $i$ is infected, \textbf{Q2A} asks us to estimate
\begin{align}
    \mathbb{P}(X(0) = x_i|O) \propto \mathbb{P}(O|X(0)=x_i)\mathbb{P}(X(0) = x_i),\;\forall i \tag{Q2A}
\end{align}
while \textbf{Q2B} asks us about the distribution of the time $t$ since the beginning of the epidemic outbreak conditioned to our observation $O$, i.e., $\PP(t \leq u|O)$ for $u \in \mathbb{R}_+$.

In subproblem \textbf{Q2C}, we aim to estimate $\mathbb{P}(X=x|O)$, where $X$ is the state of the stochastic process at the present time. However, since the number of possible network states grows exponentially with the number of nodes, it is computationally untractable to solve \textbf{Q2C}. Alternatively, we will aim to estimate the $n\times s$ marginal probabilities $\{\PP(S_i = s_j|O)\}_{i=1:n,j=1:s}$, where $S_i\in [s]$ is the current state of node $i$ and $s_j$ is one of the possible $s$ states or compartments.

To estimate solutions to questions \textbf{Q2A}, \textbf{Q2B} and \textbf{Q2C}, we are using the Monte Carlo estimator for conditioned probability, as described in the previous section. Here, $Obs_r\in \{0,1\}^k$ refers to the observation vector $O$ at detection time $t_r$ obtained in the $r$-th run of the Monte Carlo iteration with initial state $X^r(0)$. Hence, the approximate solution for \textbf{Q2A} (Patient zero detection) is the following distribution over initial states:
\begin{equation}
\hat \PP(X(0)=x_i|O) = \dfrac{|\{r\in [N_R] \colon X^r(0) = x_i \cap Obs_r= O\}|}{|\{r\in [N_R] \colon Obs_r = O\}|}.
\end{equation}
For \textbf{Q2B} (outbreak time estimation), our empirical distribution of times depends on how long it took to detect the outbreak in those runs producing an observation $Obs_r = O$, as stated in the following equation:
\begin{equation}
\hat \PP(t\leq k|O) = \dfrac{|\{r\in [N_R] \colon t_r \leq k \cap Obs_r= O\}|}{|\{r\in [N_R] \colon Obs_r = O\}|}.
\end{equation}
Defining $S_i^r$ as the status of node $i$ at the time of detection of run $r$, we have the following approximation for \textbf{Q2C} (current status assessment):
\begin{equation}
\label{q2ceq}
\hat \PP(S_i=s_j|O) = \dfrac{\{|r\in [N_R] \colon S_i^r = s_j \cap Obs_r= O|\}}{|\{r\in [N_R] \colon Obs_r = O\}|}~.
\end{equation}

\subsection{Toy example in a small network}
\label{toy2}
Consider an SIR epidemic model on the network in \cref{fig:petit} using same settings as in \cref{toy1}. Let us assume that we are continuously monitoring nodes 1 and 5. Suppose that the first time a test detects the epidemic, node 5 is infectious. Using our results, we can calculate the posterior distribution of patient-zero probabilities and the time-since-outbreak $t$, which are plotted in \cref{fig:petitzero} and \cref{fig:petittimes}, respectively. The probability of $t=0$ (detection immediately after outbreak) is 0.416, in agreement with the posterior distribution of patient zero in \cref{fig:petitzero}. The expected value of the distribution is 0.760. Finally, in \cref{fig:marginalsonly} the estimated marginal distributions for each node and state can be observed.
\begin{figure}[H]
\begin{minipage}{.45\textwidth}
  \centering
    \includegraphics[height=0.66\textwidth]{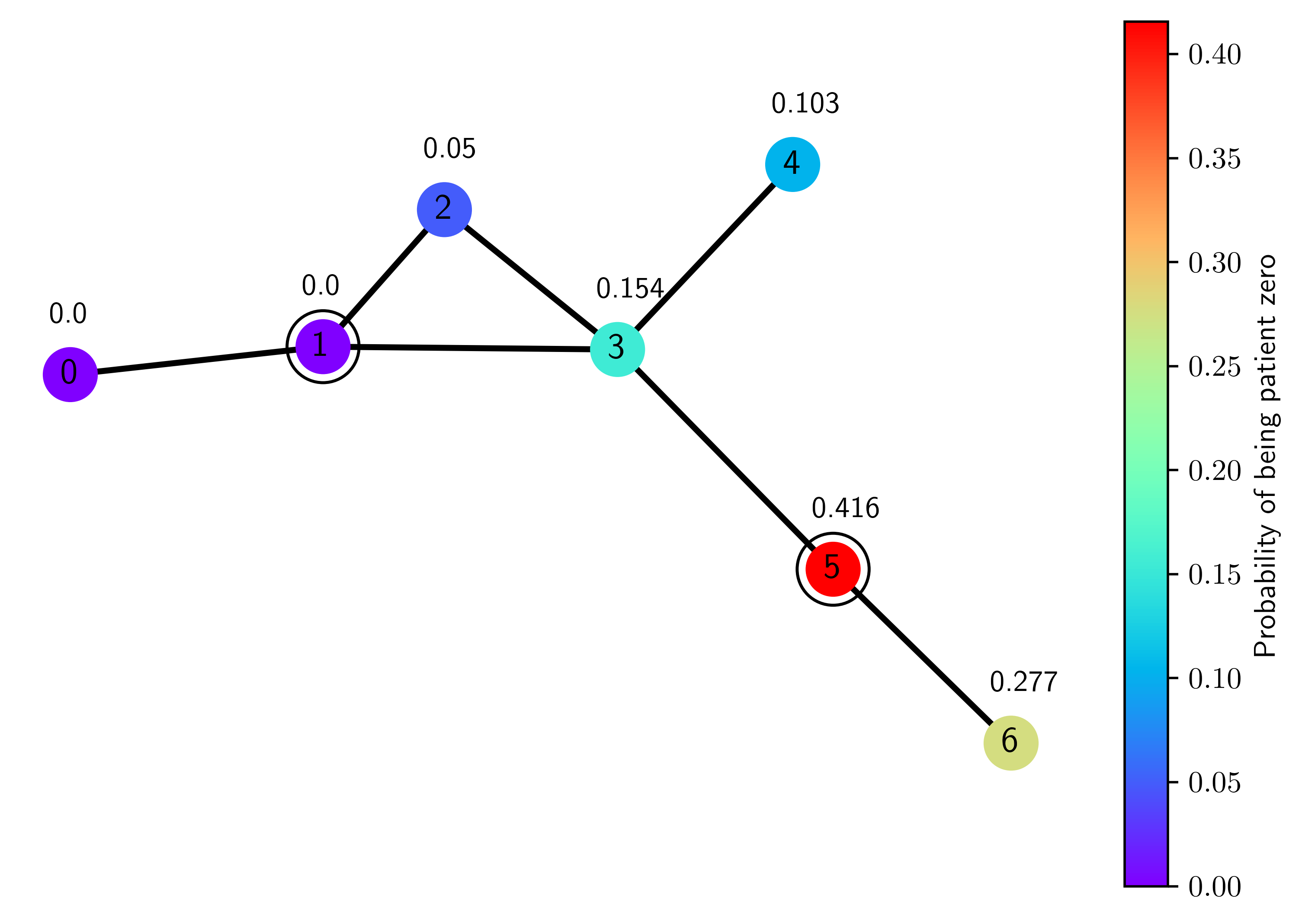}
    \caption{Posterior distribution of patient-zero probabilities for example in \cref{toy2}.}
    \label{fig:petitzero}
\end{minipage}\hfill
\begin{minipage}{.45\textwidth}
    \centering
    \includegraphics[height=0.66\textwidth]{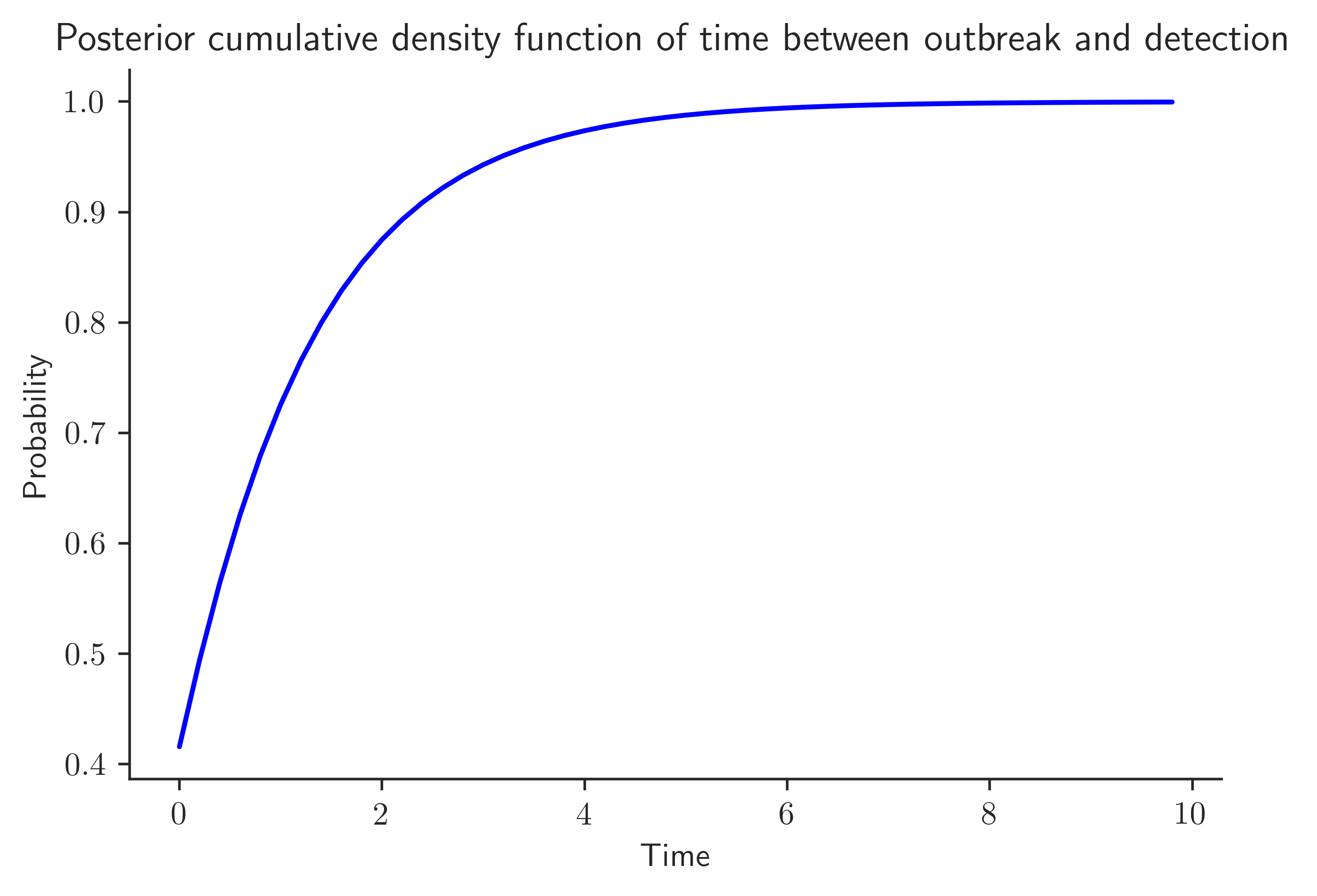}
    \caption{Posterior distribution of time-since-outbreak for the example in \cref{toy2}.}
    \label{fig:petittimes}

\end{minipage}
\end{figure}

\begin{figure}[H]
    \centering
    \includegraphics[width = 0.5\columnwidth]{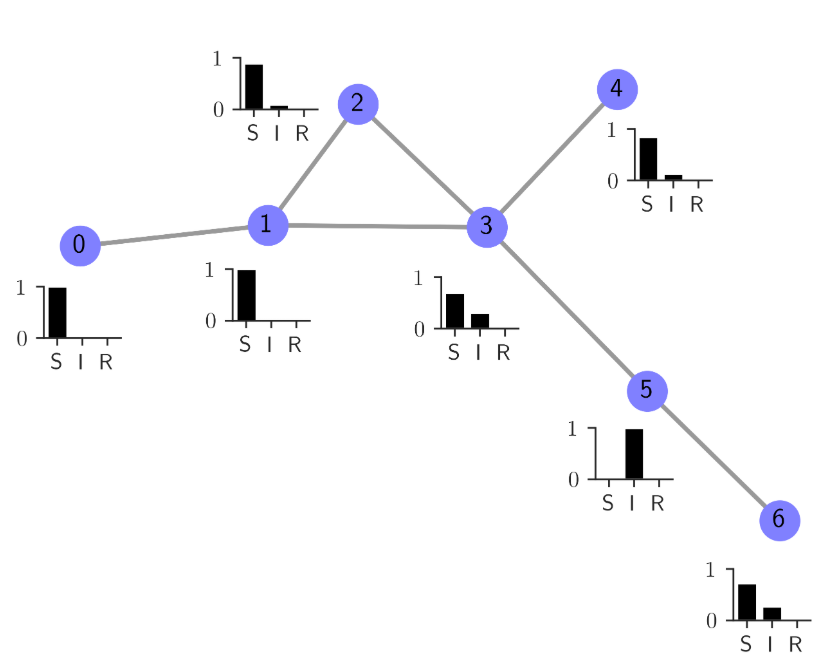}
    \caption{Marginal distributions of the current state of the network after detection, given the observation of node 5 as the first infectious between nodes 1 and 5}
    \label{fig:marginalsonly}
\end{figure}

\section{Dynamic allocation for epidemic tracking}
\label{DynamicAllocation}
In a practical scenario, once an epidemic is detected, we may be interested in retrieving as much information as possible about the state of the epidemics. In this section, we address this problem by sequentially testing nodes over time, where we have the freedom to choose the location of new tests. Let $t_0$ be the time at which an epidemic outbreak is first detected and consider that, afterwards, we are able to perform a number of new tests at times $t_1 < t_2 < \ldots < t_n$. Question \textbf{Q3}, stated below, is concerned with the design of a testing strategy aiming to maximize the amount of information extracted about the state of the disease using this series of tests.

\smallskip

\textbf{Q3 (Optimal dynamic test allocation after detection):} \emph{Assuming that we are free to sequentially allocate a number of tests at different times $\{t_i\}_{i=0}^n$, which nodes should we test at each time to maximize the `information' about the state of the disease? We assume that the tests have a known specificity $S_p$ and sensitivity $S_n$.}

\smallskip

To formalize this question, we use the classical concepts of entropy and mutual information. At the time of (first) detection, our distribution of states $X$ has an entropy of $H(X|O) = - \sum_{x\in \mathcal{S}}\PP(X=x|O)\log\PP(X=x|O)$, where $O$ is the observations at detection time. As described above, we aim to minimize the entropy of $X$ by performing a series of tests to a fixed number of nodes. Assume we test the subset of nodes $W'$. We denote by $T_{W'}$ the outcome of these $|W'|$ tests, taking values in $\{0,1\}^{|W'|}$. Thus, the average entropy of $X$ after testing $W'$ is:
\begin{align*}
H(X|O,&T_{W'}) =  \sum_{i\in\{0,1\}^{|W'|}}\PP(T_{W'}=i)H(X|O, T_{W'}=i)\\
         =& - \sum_{i\in\{0,1\}^{|W'|}}\PP(T_{W'}=i)\sum_{x\in \mathcal{S}}\PP(X=x|O, T_{W'}=i)\log\PP(X=x|O, T_{W'}=i). \numberthis~
\end{align*} 
The problem we aim to solve is the maximization of the mutual information
\begin{align}
\argmax_{W'\subset V} I(X|O;T_{W'}) = \argmax_{W'\subset V} H(X|O) - H(X|O,T_{W'}) ~.
\end{align}
As mentioned before, computing the full distribution $\PP(X=x|O)$ is unfeasible due to the exponential size of the probability space;hence, we propose below an approximation to this problem using the information we have from solving previous questions. Assume that we have access to an estimation of the marginals $\{\PP(S_i = s_j|O)\}_{i,j}$. Since we have no more information about the distributions, we can assume $\{S_i\}_{i=1}^n$ are independent to obtain
$\PP(S_1 = s_1, S_2 = s_2, \ldots, S_n = s_n | O) = \prod_{i=1}^n\PP(S_i = s_i|O)~. $
Making the (approximate) assumption that the testing process does not destroy this independence, we obtain
\begin{align}
&H(X|O, T_{W'}) = H((S_1, S_2, \ldots, S_n)|O, T_{W'}) = \sum_{i\in W'}H(S_i|O, T_{W'}) + \sum_{i\not\in W'}H(S_i|O)~,
\end{align}
where the last equality comes from the fact that a test on $W'$ does not change our knowledge of other nodes under the independence assumption. Hence,  
\begin{align}
H(X|O) - H(X|O,T_{W'}) =& \sum_{i=1}^{n}H(S_i|O) - \sum_{i\in W'}H(S_i|O, T_{W'}) - \sum_{i\not\in W'}H(S_i|O) \\ 
=& \sum_{i\in W'}(H(S_i|O)- H(S_i|O, T_{\{i\}})) \numberthis~.
\label{defMI1}
\end{align}
This means that we can effectively rank the nodes, with a score equal to $H(S_i|O)-H(S_i|O, T_{\{i\}}) \geq 0$, and include in $W'$ the desired number of nodes of the highest score. This score can be computed with our estimated marginals from \textbf{Q2C}, as follows. On the one hand, by the definition of entropy, we have that
\begin{equation}
\label{defMI2}
H(S_i|O) = -\sum_{j=1}^{s} \PP(S_i = s_j|O)\log\PP(S_i=s_j|O)~.
\end{equation}
On the other hand, $H(S_i|O, T_{\{i\}})$ depends on the capacity of the tests to distinguish states (determined by their sensitivity and specificity). In particular, we consider (binary) tests with known specificity $S_p$ and sensitivity $S_n$. We denote by $+, -$ the positive and negative results in one such test, and $S_i \in D$ the event that node i is a detectable state. Under these assumptions, the expected entropy after the test is given by:
\begin{equation}
\label{defMI3}
H(S_i|O, T_{\{i\}}) = \PP(T_{i} = +) H(S_i|O, T_{\{i\}}=+) +  \PP(T_{i} = -) H(S_i|O, T_{\{i\}}=-)
\end{equation}
We can calculate each one of the terms in this expression as follows: $\PP(T_{i} = +) = S_n\PP(S_i \in D) + (1-S_p)(1-\PP(S_i \in D))$ and $\PP(T_{i} = -) = (1-S_n)\PP(S_i \in D) + S_p (1-\PP(S_i \in D)) = 1-\PP(T_{i} = +)$. Also, $H(S_i|O, T_{\{i\}}=+)$ and $H(S_i|O, T_{\{i\}}=-)$ are the entropies of Bernoulli variables with parameters $\PP(S_i \in D|T_{i} = +)$ and $\PP(S_i \in D|T_{i} = -)$. Finally, $\PP(S_i \in D|T_{i} = +) = \frac{S_n\PP(S_i \in D)}{\PP(T_i=+)}$ and  $\PP(S_i \in D|T_{i} = -) = \frac{(1-S_n)\PP(S_i \in D)}{\PP(T_i=-)}$ from a direct application of Bayes' Formula. \\

The computation of $H(S_i|O)- H(S_i|O, T_{\{i\}})$ reveals that, if the test deviates from the ideal test with $S_p = S_n = 1$, then the optimal distribution to test deviates from the uniform $\PP(+)=\PP(-)=0.5$. The optimal $\PP(+)$ solves a transcendental equation as a function of $S_p$ and $S_n$. As an example, if $S_n = 0.957$ and $S_p = 0.99$ (as obtained from real data in the Autobio Diagnostics Co. RDT IgM Covid-19 tests \cite{GlobalPr97:online}), the probability of testing positive before the test resulting in a maximal mutual information is $\PP(+) = 0.4836$. Using this procedure we can rank the nodes and choose the top $k_i$ scores to allocate $k_i$ tests at the time instance $t_i$. The full procedure alternates between testing and simulations to update the current status of the network and dynamically decide on subsequent tests. After tests at time $t_i$ are taken, we define $D_i$ as the distribution over nodes conditioned on the test results. If the process is Markov, the conditions needed to start the model are just the status of all nodes. For non-markovian processes, other elements, such as how much time a node has been in its status, need to be considered. We include these under $D_i$, understanding that we sample all the values needed to uniquely determine the system current status and evolution. We use the Monte Carlo simulator with initial states sampled from $D_i$ until $t = t_{i+1}-t_i$. After we have enough runs, one can estimate the marginal distributions using \cref{q2ceq}, then use \cref{defMI1}, \cref{defMI2}, and \cref{defMI3} to decide where the new $k_{i+1}$ tests are allocated, and then calculate $D_{i+1}$ using the results from the test. This procedure is conceptually similar to particle filtering or sequential Monte Carlo methods, in which measurements of reality are combined sequentially with a simulator of the associated dynamics.

\subsection{Toy example in a small network}

\cref{fig:petitmarginals} shows the process of adaptive testing in the case of the toy example in \cref{fig:petit}. We start with the marginals in \cref{fig:marginalsonly}, which are used to calculate their entropies and decide the optimal nodes to test next. After the test is taken, one uses the information to update the marginals and use them to sample the initial condition for subsequent Monte Carlo runs, in which we simulate the stochastic process until the time in which we are allowed to allocate more tests (e.g., every week). After that, we can update the marginals, decide on next tests to take, and repeat the cycle.
\begin{figure}[H]
    \centering
    \includegraphics[width = 0.95\textwidth]{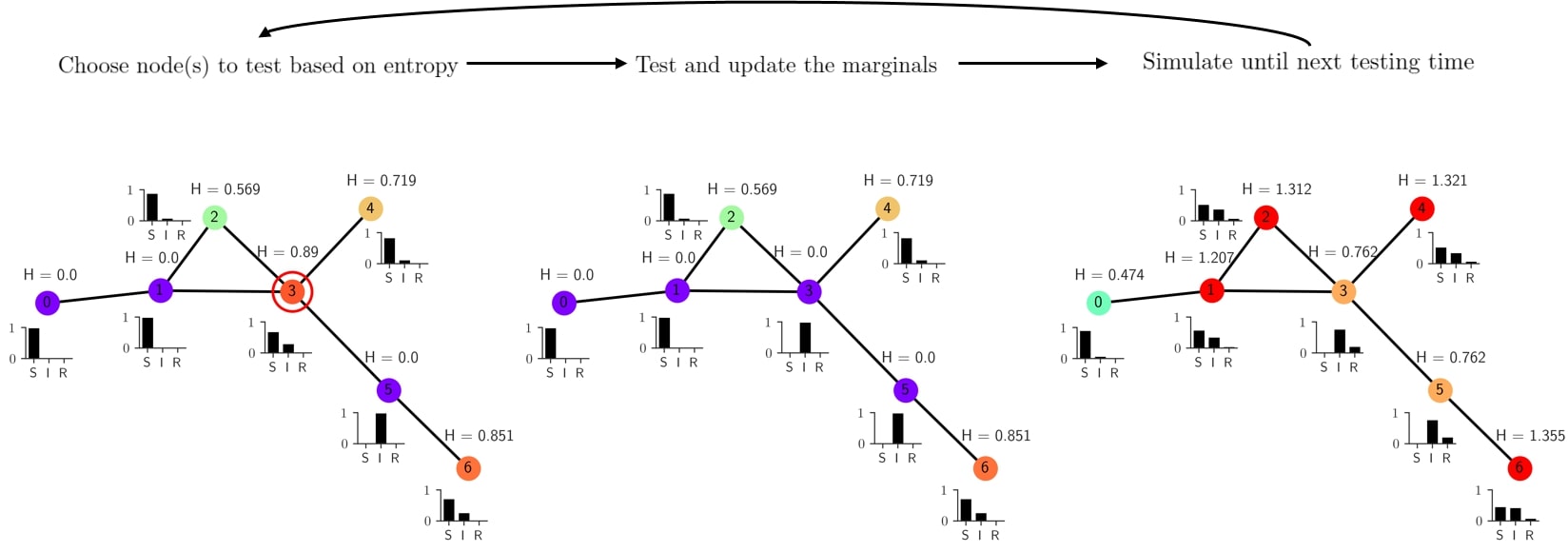}
    \caption{Dynamic testing process, consisting on alternating testing and simulation, applied to the toy example network}
    \label{fig:petitmarginals}
\end{figure}
\section{Experiments}
\label{Experiments}
In this section we illustrate the proposed procedures to a non-markovian model of Covid-19 in a real human interaction network. In order to simulate the stochastic processes, we use an event-driven simulation algorithm in which the next events (infections, recoveries, etc.) are stored in a priority queue and processed in order of time \cite{Kiss2017}. This method can be used to efficiently simulate many stochastic models such as the model studied here. We use the Hypertext 2009 network \cite{konect:2016:sociopatterns-hypertext}, a network of human-to-human interactions, in our simulations. The ACM Conference on Hypertext and Hypermedia 2009 was held in Turin, Italy in 2009 and, during the conference, the conference badges included Radio-Frequency Identification (RFI) devices able to mine face-to-face proximity relations \cite{Isella2011}. The exchange of radio packets between badges implies a proximity of less than 1-1.5 $m$, a distance in which contagious diseases could spread. In this network, a node represents a conference visitor and an edge represents a face-to-face contact that was active for at least 20 seconds. The network has $n = 100$ nodes and $m = 946$ edges once we aggregate edges over time during the first day of the conference.

For this network, we use an adaptation to networks of a realistic non-markovian model of Covid-19 proposed in \cite{Ferretti2020}. In this work, the authors infer that, for Covid-19, the incubation period (time between contracting the disease and showing symptoms) follows a lognormal distribution with meanlog 1.644 and sdlog 0.363, and the generation time (time between infection of the source and infection of the tarauthor follows a Weibull distribution with shape parameter 2.826 and scale parameter 5.665. Additionally, the authors infers that the proportion of infectious individuals who are asymptomatic is 0.4, and that asymptomatic transmission rate is 10 times lower than for symptomatic patients. We use this data to create a non-markovian model with susceptible, presymptomatic, symptomatic, asymptomatic, and removed compartments in which each node can independently infect its neighbours as long as it is not susceptible nor removed. The times for that infection to occur and symptoms to appear is sampled from the distributions in \cite{Ferretti2020}. We draw the random time to full recovery from first infection to removal from a normal distribution of mean 14 and standard deviation 2, both for symptomatic and asymptomatic carriers. The model is summarized in \cref{fig:Covidmodel}.
\begin{figure}[!htbp]
    \centering
    \includegraphics[width = 0.7\columnwidth]{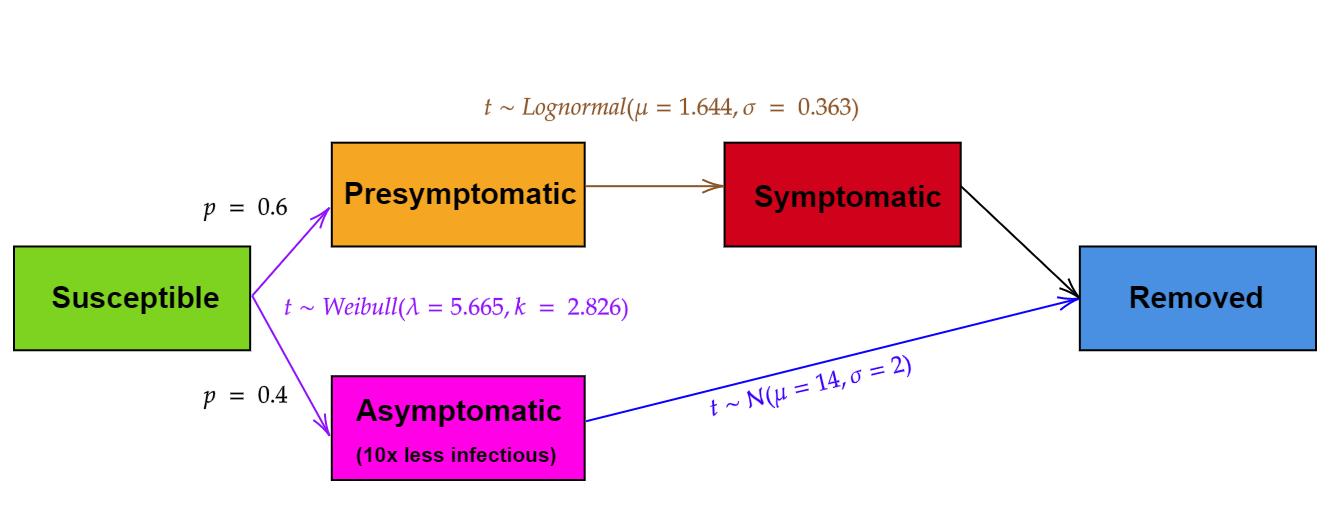}
    \caption{Summary of the non-markovian model of Covid-19 spreading in networks based in \cite{Ferretti2020}}
    \label{fig:Covidmodel}
\end{figure}

We assume that the outbreak is started by a single infectious node chosen uniformly at random. We then monitor $k = 10$ nodes decided according to the greedy scheme in \cref{alg:greedy1}, which aims to maximize the probability of detection during the first $\tau = 3$ days of the outbreak. We test the greedy algorithm against three simple baselines: uniformly random node subset selection, random node subset selection weighted by node degree, and random node selection eliminating neighbors from chosen nodes iteratively. \cref{fig:distrgreedy} summarizes the results, where the greedy algorithm outperforms all the strategies by around 3$\%$ in probability. In order to understand how this translates to real scenarios in practice, we run $10^5$ simulations for the greedy test placement and $10^5$ for the best randomly found placement in which we set lockdown measures as soon as the epidemic is detected in each case. The curves of infectious (asymptomatic + symptomatic + presymptomatic) and recovered nodes for each case can be seen in \cref{fig:CVcurves}. On average, 4 nodes out of the 100 nodes do not contract the disease by using the greedy placement instead of the best random placement. 

\begin{figure}[H]
\begin{minipage}{.45\textwidth}
  \centering
    \includegraphics[width=.95\columnwidth]{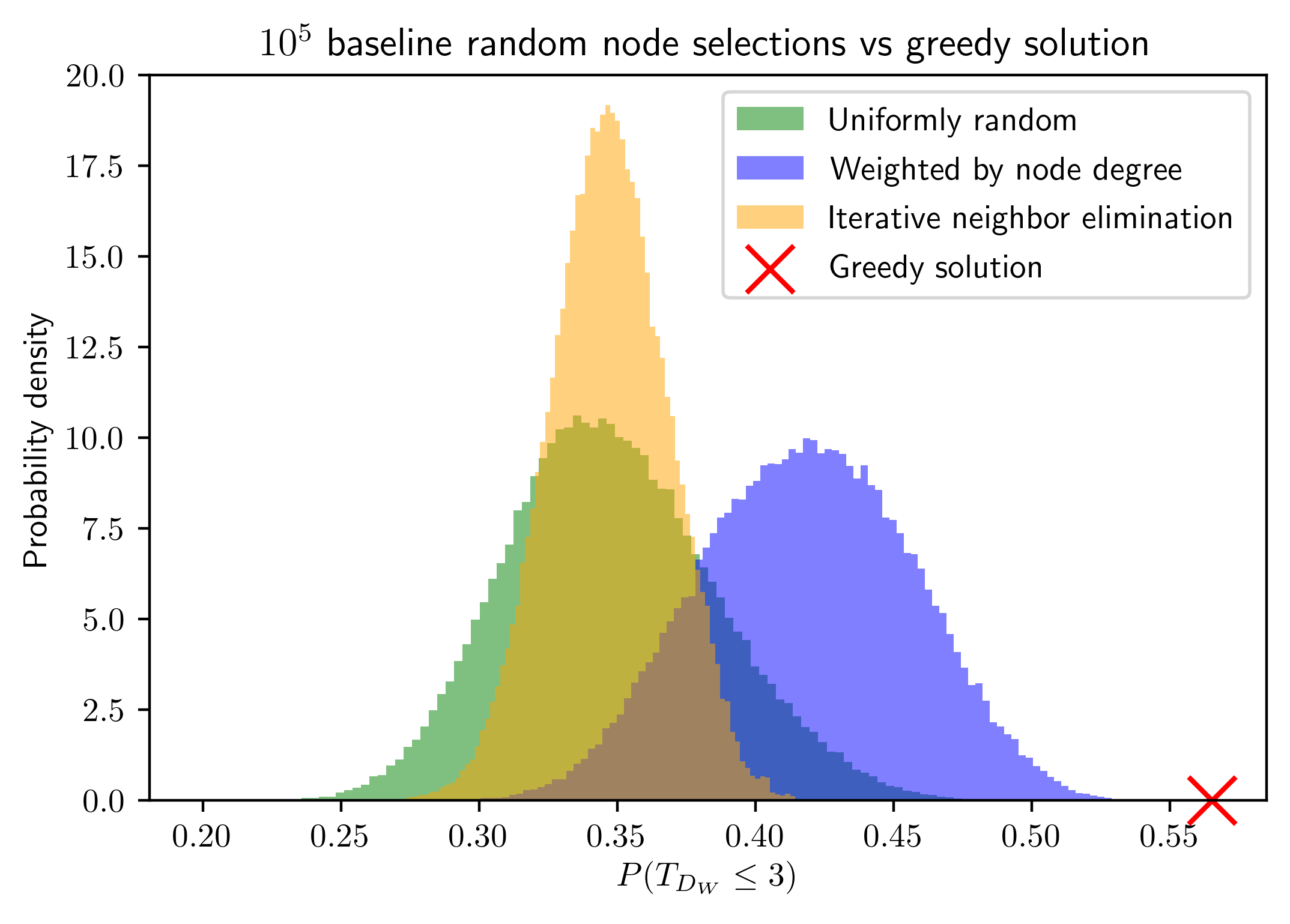}
    \caption{Histogram of results after using the three baseline random placement strategies in comparison to the greedy algorithm in order to place $k = 10$ tests. The greedy algorithm scores a detection probability of 0.57 while the best random solution scores 0.542}
    \label{fig:distrgreedy}
\end{minipage}\hfill
\begin{minipage}{.45\textwidth}
    \centering
    \includegraphics[width=.95\columnwidth]{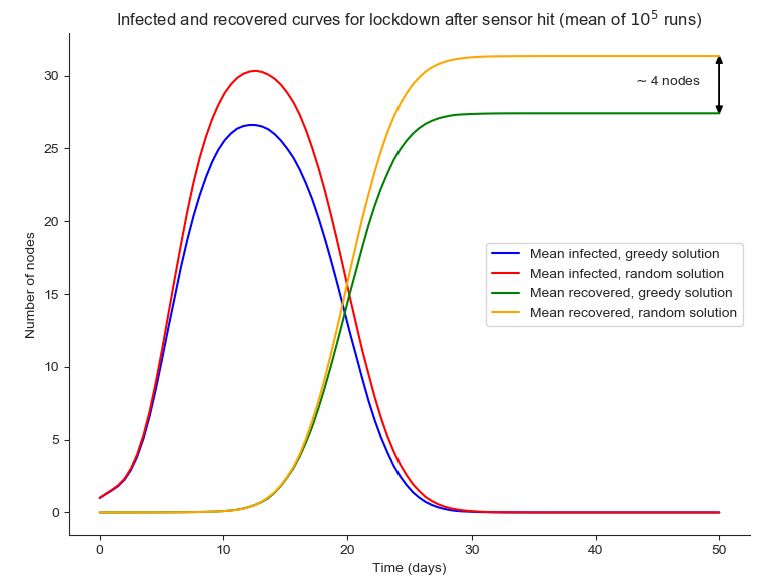}
    \caption{Infectious (asymptomatic + symptomatic + presymptomatic) and recovered mean curves for $10^5$ runs in which lockdown is imposed once detecting the epidemic. The small probabililty gain of using the set of sensors found with the greedy strategy translates to 4 less nodes being infectious on average}
    \label{fig:CVcurves}

\end{minipage}
\end{figure}


A similar analysis as the one performed in the toy example in \cref{toy2} can be performed to estimate patient zero and results can be seen in \cref{fig:pat0covid}. The most likely node is in this case the node in which the epidemic was detected, which is placed in the middle of the graph in \cref{fig:pat0covid}. Similarly as before, we can estimate the probability density function of time-since-outbreak at the time of detection. Here, we do it for three different kinds of tests: tests that detect antibodies (meaning all kind of non-susceptible nodes), `tests' that detect symptoms only, and `tests' that detect removed people only. These two last cases correspond to the cases in which epidemic outbreaks are detected late instead of using actual viral tests, simulating scenarios in which countries or populations are unprepared for an outbreak and can only detect it after the first death (or person with symptoms) is detected. The results can be seen in \cref{fig:timedistrs}, in which we can observe that the difference is of the order of several days in each case.

\begin{figure}[H]
\begin{minipage}{.45\textwidth}
    \centering
    \includegraphics[width=\columnwidth]{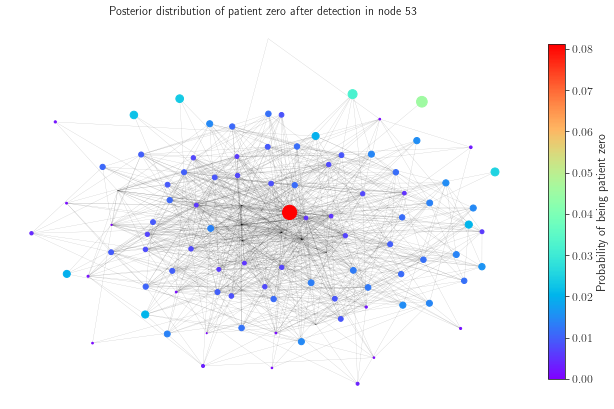}
    \caption{Probability distribution of being patient zero, after detection in node 53 (in the center of the graph) is taken into account. Node size is proportional to probability}
    \label{fig:pat0covid}\end{minipage}\hfill
\begin{minipage}{.45\textwidth}
    \centering
    \includegraphics[width=\columnwidth]{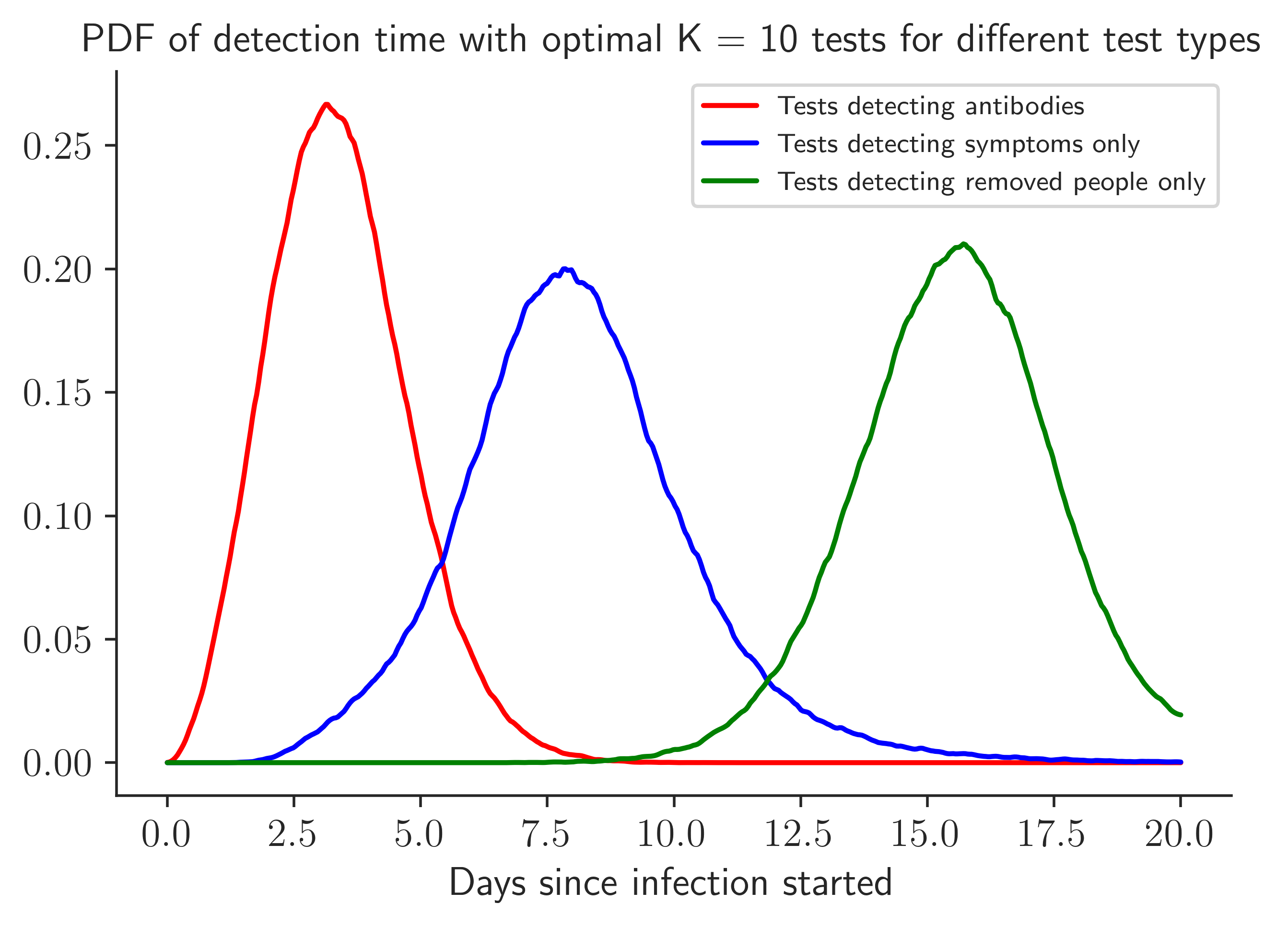}
    \caption{Probability density functions of times between disease outbreak and detection for optimal monitorization of $k = 10$ nodes using tests of different type, corresponding to being able to test people, detecting just the symptoms or just the death of patients}
    \label{fig:timedistrs}
\end{minipage}
\end{figure}

Finally, we illustrate the proposed adaptive testing algorithm with a fixed amount of tests at $t=t_0$ and every 3 days, up to 4 times. We compare it with a baseline algorithm of randomly selecting which nodes to choose at each iteration. In this comparison, testing is only used to monitor the epidemic (i.e no lockdown measures are imposed regardless of test results). In \cref{fig:Dynamictesting}, the classification accuracy compared to the real run (taking the most likely class in the estimated marginals as the prediction) and mean entropy of the predicted marginal distributions are plotted over time, averaging over $10^4$ different real runs. As expected, more tests translates to higher accuracy and distributions with less entropy. By comparing testing strategies, it can be observed how the proposed strategy maintains the classification accuracy but improves on the uncertainty that the distributions convey. The time of higher uncertainty is around 6 days after detection, as there are more possible scenarios of the current state of the pandemic than later on. The reason for this is that, in this non-lockdown scenario, after a certain point most of the nodes will most likely have been infected.
\begin{figure}[H]
    \centering
    \includegraphics[width = 0.9\columnwidth]{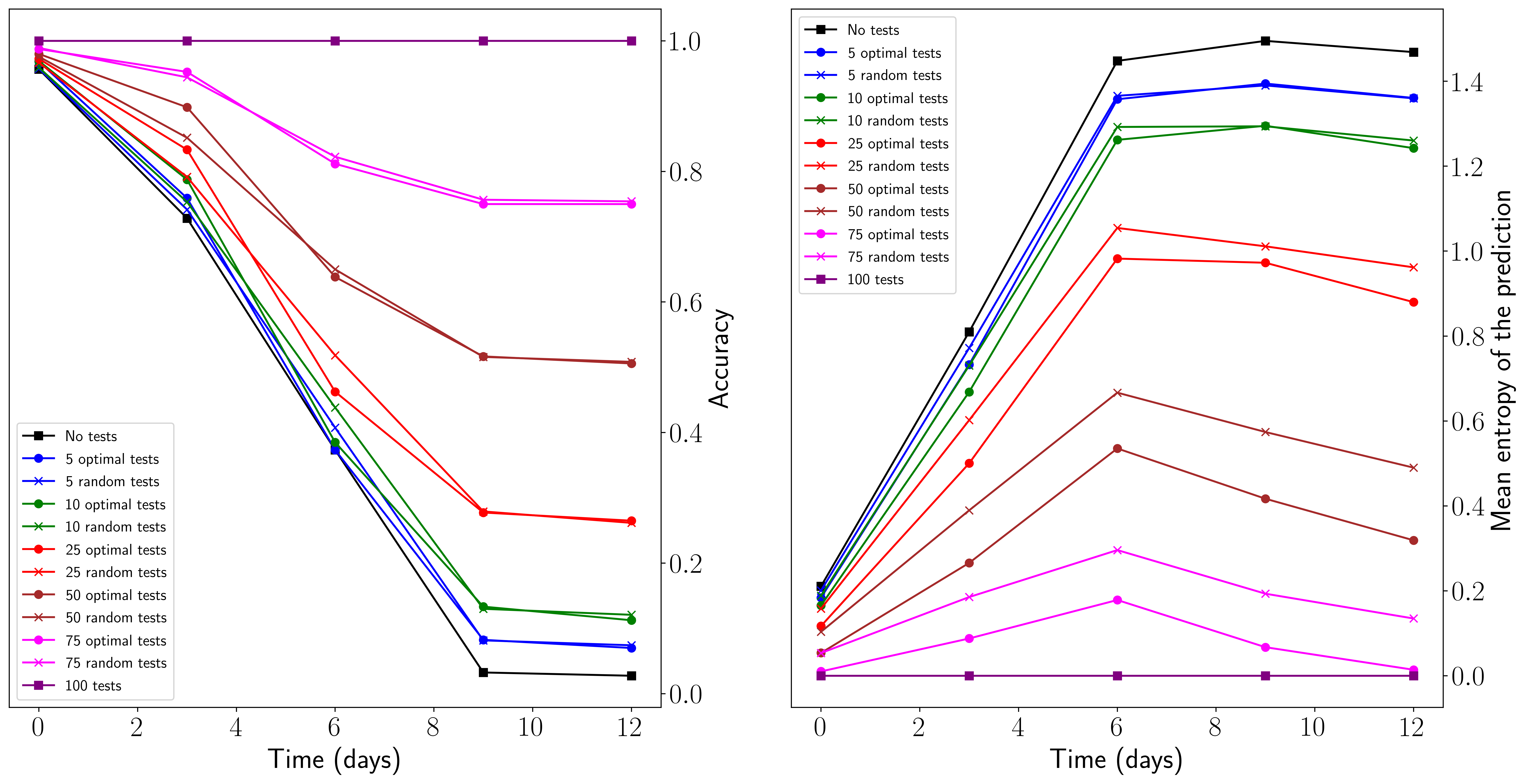}
    \caption{Accuracy and mean entropy time evolution in the dynamic testing scenario in which a fixed number of tests are used every 3 days, up to 12 days after the outbreak is detected. Averages of $10^4$ runs}
    \label{fig:Dynamictesting}
\end{figure}



\section{Conclusions}

We have introduced a flexible framework to analyze problems concerning the early detection of epidemic outbreaks, as well as the dynamic allocation of tests to maximize the information retrieved about the state of the infection. The proposed framework is able to handle any discrete-time compartmental stochastic process and tests with a given specificity and sensitivity. We have stated and solved several problems of practical interest by analyzing continuous-time stochastic compartmental models over complex networks. First, we have considered the problem of designing a monitoring system whose objective is to detect a novel epidemic outbreak as soon as possible. In particular, we have developed an algorithm able to select a subset of individuals to be continuously monitored in order to detect the onset of an epidemic as fast as possible. We have mathematically described this problem as a hitting-time probability maximization and use submodularity optimization techniques to derive explicit quality guarantees for the proposed solution.
Second, assuming that an epidemic outbreak has been detected, we have also considered the problem of dynamically allocating viral tests over time in order to maximize the amount of information gained about the state of the epidemic. We  have proposed an adaptive allocation strategy with quality guarantees based on the concepts of information entropy and mutual information. For all these problems, we have derived analytical solutions for Markovian stochastic compartmental models, as well as efficient Monte-Carlo-based algorithms for non-Markovian dynamics and large-scale networks. We have illustrated the performance of the proposed algorithms using numerical experiments involving a model of Covid-19 applied to a real human contact network.

\appendix
\section{Analytical solution for Markov chains}
\label{Markovchains}
In this section we provide analytical solutions for the questions from \cref{sec:intro} in the case the Markov Property holds and the epidemic model defines a continuous-time Markov chain. 
Similarly as in the general case, the epidemic model can be formulated as a continuous time Markov chain of state space $\mathcal{S}$ with $|\mathcal{S}| = s^n$ if the Markov property holds. An initial probability distribution $D$ such that $X(0) \sim D$ is also assumed. The continuous-time Markov chain is characterized by a transition rate matrix $Q$ of dimensions $s^n \times s^n$.

\subsection{Early detection of epidemic outbreaks with limited viral tests}

We use the same notation as in the general case: $T_{D_W}$ is the minimum time in which the Markov chain reaches the detection set of $W \subset V$. We can analytically calculate $\mathbb{P}\left(T_{D_W} \leq \tau \right) = \mathbb{P}\left(T_{D_W} \leq \tau \right | T_{D_W} < \infty) \mathbb{P}\left(T_{D_W} < \infty\right)$, where $\mathbb{P}\left(T_{D_W} < \infty\right)$ refers to the probability of the Markov Process ending in an absorbing state in $D_W$. Some models, such as SIR, have absorbing states which represent the epidemic "dying" before any detection is done (for example, the first person gets cured before transmitting it to anyone), which might therefore not be part of $D_W$. To treat that, we use the jump matrix of the Markov chain. If $Q$ is the transition rate matrix of a continuous-time and discrete space-Markov chain and we are at state $i$, the probability of the next jump of the Markov chain being to state $j \neq i$ is $\frac{Q_{ij}}{-Q_{ii}}$. This lets us define the jump matrix $M$, where $M_{ij} = \dfrac{Q_{ij}}{-Q_{ii}}$ if $i \neq j$ and $M_{ii} = 1-\sum_{i\neq j}M_{ij}$, which corresponds to the matrix of a discrete time Markov chain describing the jumps the Markov chain makes without taking into account how long it takes to do such jumps. For calculations such as the probability of being absorbed in a particular state, we can use the jump matrix and theory from discrete-time absorbing Markov chains. In particular, we can calculate $\PP\left(T_{D_W} < \infty\right)$ from the fundamental matrix of the jump matrix created from $Q$.

\paragraph{Conditional absorption Markov chain theory}
To calculate $\mathbb{P}\left(T_{D_W} \leq \tau \right | T_{D_W} < \infty)$ we need to eliminate the absorbing states not in $D_W$ from the chain, which we denote $S^-$. The dynamics of those runs that do not end in $S^-$ are Markovian, and its rate transition matrix can be found according to the following proposition.

\begin{prop}
Given a Markov chain with two sets of absorbing states $S^+$ and $S^-$ and a rate transition matrix $Q$, one can construct a matrix $Q^+$ corresponding to the Markovian dynamics of those processes that get absorbed at $S^+$ (which we can write as $X_\infty \in S^+$) as follows:
\begin{equation}
Q^+_{ij} = 
\begin{cases}
 Q_{ij} \dfrac{\mathbb{P}\left(X_\infty \in S^+|X(0) = j\right)}{\mathbb{P}\left(X_\infty \in S^+|X(0) = i\right)} \quad &\mathbb{P}\left(X_\infty \in S^+|X(0) = i\right) \neq 0 \\
 0 \quad &\mathbb{P}\left(X_\infty \in S^+|X(0) = i\right) = 0
\end{cases}
\end{equation}
Furthermore, the initial probability distribution $\{\PP(X(0) = i)\}_i$ also gets modified: 
\begin{equation}
\PP(X(0) = i|X_\infty \in S^+) = \dfrac{\PP(X_\infty \in S^+|X(0) = i)\PP(X(0) = i)}{\sum_j \PP(X_\infty \in S^+|X(0) = j)\PP(X(0) = j)}
\end{equation}
\end{prop}
\begin{proof}
The initial distribution modification is a direct application of Bayes' Theorem. For a state i such that $\PP(X_\infty \in S^+|X(0) = i) = 0$, then $\PP(X(0) = i|X_\infty \in S^+) = 0$, and $\PP(X(t) = i|X_\infty \in S^+) = 0$ must also hold for all $t > 0$. This means that all $Q_{ki}$ for $k \neq i = 0$, and therefore state $i$ will not be reached at all in the new dynamics. We can therefore safely remove such states from the chain, or assign its transitions rate a value of 0 so that they act as absorbing non-achievable states. For a state $i$ such that $\PP(X_\infty \in S^+|X(0) = i) > 0$, remembering that the matrix $Q$ satisfies
\begin{equation*}
    Q_{ij} = \lim_{h \rightarrow 0} \dfrac{\PP(X(h) = j|X(0) = i) - \delta_{ij}}{h}
\end{equation*}
we can derive that
\begin{align*}
    Q^+_{ij} &= \lim_{h \rightarrow 0} \dfrac{\PP(X(h) = j|X(0) = i, X_\infty = S^+) - \delta_{ij}}{h} \\
    &= \lim_{h \rightarrow 0} \dfrac{\dfrac{\PP(X(h) = j|X(0) = i)\PP(X_\infty = S^+|X(h) = j, X(0) = i)}{\PP(X_\infty = S^+|X(0)=i)} - \delta_{ij}}{h} \\
    &= \lim_{h \rightarrow 0} \dfrac{\dfrac{\PP(X(h) = j|X(0) = i)\PP(X_\infty = S^+|X(0) = j)}{\PP(X_\infty = S^+|X(0)=i)} - \delta_{ij}}{h} \\ 
    &= \dfrac{\PP(X_\infty = S^+|X(0) = j)}{\PP(X_\infty = S^+|X(0)=i)} \lim_{h \rightarrow 0} \dfrac{\PP(X(h) = j|X(0) = i) - \delta_{ij}}{h} \\ 
    &= \dfrac{\PP(X_\infty = S^+|X(0) = j)}{\PP(X_\infty = S^+|X(0)=i)} Q_{ij}
\end{align*}
where we used $\delta_{ij}\frac{\PP(X_\infty = S^+|X(0) = i)}{\PP(X_\infty = S^+|X(0)=j)} = \delta_{ij}$ 
\end{proof}

In practice, the quantities $\{\mathbb{P}\left(X_\infty \in S^+|X_0 = i\right)\}_i$ can be found from the fundamental matrix of the jump matrix 

Once we have a Markov chain with the only set of absorbing states $S^+$ (and therefore $\PP(X_\infty \in S^+) = 1$) the distribution of stopping times (usually called hitting times in the context of Markov chains) to $S^+$ follows a phase-type distribution. If we collapse all the states of $S^+$ into one (by adding the probability rates that reach it), the hitting times are unchanged and the rate matrix of the Markov chain with $N_t$ transient states ($N_t$ depends on the choice of compartment model and $n$) takes the form $
Q' = \left(\begin{array}{cc}
     0 & 0  \\
     S^0 & S
\end{array}\right)$, where $S$ is a $N_t\times N_t$ matrix and $S_0$ is equal to $-S\Vec{1}$, where $\Vec{1}$ is the column vector of all ones. The time it takes to be absorbed in state 0 starting from a vector of initial probabilities $\Vec{\alpha}$ is distributed according to the distribution function $
F(t) = \PP(T_{S^+} \leq t) = 1- \Vec{\alpha}\exp(St)\vec{1}$, where exp($\cdot$) denotes the matrix exponential. The expected value is $-\Vec{\alpha}S^{-1}\Vec{1}$.

The full equation reads

\begin{align*}
    \PP(T_{D_W} \leq \tau) &= \PP(T_{D_W} < \infty)\PP(T_{D_W} \leq \tau | T_{D_W} < \infty) \\
    &= \PP(T_{D_W} < \infty)(1- \Vec{\alpha}\exp(St)\vec{1}) \numberthis~.
\end{align*}
where $S$ is obtained by the process of first conditioning and then collapsing the absorbing states and $\alpha$ also comes from conditioning $D$ and then collapsing the states in $D_W$.

However, this is infeasible in practice as $S$ scales roughly as $Q$, which is $s^n \times s^n$.

\subsection{Estimation of the state of the disease}
We now continue to solve analytically the rest of the tasks. \textbf{Q2A} and \textbf{Q2B} ask about patient zero posterior probabilities and outbreak time estimation. 


We define a state of the Markov chain to be \emph{compatible} with our observation if in that state the tested nodes are in a state which agrees with the tests. These include all possibilities for non-tested nodes but may include some variations in tested-nodes if the tests do not perfectly distinguish all states. We let $\mathcal{C} \subset D_W \subset \mathcal{S}$ be the set of compatible states, and $O$ denote our observation. 

For \textbf{Q2A} we can apply Bayes' Theorem
\begin{align}
    \mathbb{P}(X_0 = x|O) = \frac{\mathbb{P}(O|X_0=x)\mathbb{P}(X_0 = x)}{\mathbb{P}(O)} \propto \mathbb{P}(O|X_0=x)\mathbb{P}(X_0 = x)~,
\end{align}
as $\PP{(O)}$ is just a constant that ensures $\sum_{i\in\mathcal{S}}\mathbb{P}(X_0 = i|O) = 1$. We know $\mathbb{P}(X_0 = x)$, as the initial distribution $D$ is known. For $\mathbb{P}(O|X_0=x)$, we again consider all states of the detection set of the placed tests as absorbing, and we need to sum the probabilities of getting absorbed to exactly those states in the detection set which are compatible with our observation. Therefore
\begin{align}
\mathbb{P}(O|X_0=x) = \sum_{\alpha \in \mathcal{C}}\PP(X_\infty = \alpha|X_0=x)~.
\end{align}
The probability of ending in a specific absorbing state starting from a specific transient state can be found with the fundamental matrix of the jump Markov chain matrix. \textbf{Q2B} asks about the distribution of time since the epidemic began. To calculate the distribution function $\PP(t\leq k|O)$ conditioned on the absorption happening on a compatible state we can do exactly the same as we have done to calculate $\mathbb{P}\left(T_{D_W} \leq \tau \right | T_{D_W} < \infty)$ except for replacing $D_W$ for its subset $C$. \textbf{Q2C} asks about the probability distributions of the current state over the states in $D_W$. This is calculated using the same idea as \textbf{Q2A}. We denote $X$ the actual state

\begin{align}
\mathbb{P}(X=x|O) = \frac{\mathbb{P}(O|X=x)\mathbb{P}(X=x)}{\mathbb{P}(O)} \propto \mathbb{P}(O|X=x)\mathbb{P}(X=x) = \mathbb{I}(x\in\mathcal{C})\mathbb{P}(X=x)~.
\end{align}

$\mathbb{P}(X=x)$ is the probability of being absorbed at state $x$, which is known a priori with the fundamental matrix. Therefore, we see that the observation just restricts the probability distribution from $D_W$ to its subset $\mathcal{C}$. We can now solve for the probability of node $i$ being in state $s_l$ by summing over the posterior probabilities of all states in which $i$ is in $s_l$.

\subsection{Dynamic allocation for epidemic tracking}


Note that if we perfectly know $P(X=x|O)$, then we are able to evaluate the expression for $H(X|O,T_{W'})$, since if $\mathcal{C}_i 
\subset \mathcal{C}$ are the compatible states with test output $t_i$, $\PP(T_{W'}=t_i) = \sum_{\alpha \in \mathcal{C}_i}\PP(\alpha|O)$
and similarly as before,
\begin{align}
\PP(X=x|O, T_{W'}=t_i) \propto \PP(T_{W'}=t_i|O, X=x)\PP(X=x|O) = \mathbb{I}(x\in \mathcal{C}_i)\PP(X=x|O)~.
\end{align}
We can now evaluate the mutual information for all tests to obtain the best one. Therefore, in this case we can perform the following iterative procedure:

\begin{itemize}
    \item At $t = t_i$, we test nodes according to the entropy criterion. After the tests we update the distributions conditioned on test results using Bayes' theorem. Let $D_i$ be the distribution over nodes conditioned on the test results
    \item Run the Markov chain analytically until $t = t_{i+1}-t_i$ with initial distribution $D_i$
    \item Calculate the state distributions at time $t_{i+1}$
    \item Decide which nodes to test at time $t_{i+1}$ according to the criteria in Q5, and calculate $D_{i+1}$ with the obtained results
\end{itemize}

\section{Proof of \cref{thm:submod} and \cref{thm:submodfhat}}
\label{Proofs}
\subsection{Proof of \cref{thm:submod}}
\begin{theorem}[Submodularity of $f_\tau$]
\label{thm:submodbi}
$f_\tau \colon W \mapsto \mathbb{P}\left(T_{D_W} \leq \tau\right)$ is a submodular function.
\end{theorem}
\begin{lemma}
Let $S$ be a finite set and consider $M$ a continuous time stochastic process over the states of $S$. Then, for $\tau \in \mathbb{R}_+$ the function $h(W) \colon W \mapsto \mathbb{P}\left(T_W \leq \tau\right)$, where $W\in\mathcal{P}(S)$ is submodular.
\label{teo1}
\end{lemma}
\begin{proof}
We want to see that for $X\subset Y$ and $x \in S\setminus Y$, $
\mathbb{P}(T_{X\cup\{x\}} \leq \tau) - \\ \mathbb{P}(T_{X} \leq \tau) \geq \mathbb{P}(T_{Y\cup\{x\}} \leq \tau) - \mathbb{P}(T_{Y} \leq \tau)$. For $Z \subset S$, $\mathbb{P}(T_{Z\cup\{x\}} \leq \tau) - \mathbb{P}(T_{Z} \leq \tau) = \PP(T_{\{x\}} \leq \tau \cap T_Z > \tau)$. But since $X\subset Y$, $T_Y \geq \tau \implies T_X \geq \tau$, and so $\PP(T_{\{x\}} \leq \tau \cap T_X \geq \tau) \geq \PP(T_{\{x\}} \leq \tau \cap T_Y \geq \tau)$ and we are done.
\end{proof}

The second part of the submodularity proof for $f_\tau$ concerns being able to conserve the submodularity of $h$ under composition with functions of certain properties.
\begin{lemma} (Conservation of submodularity under pullback). Let $V, S$ be finite sets, and let $h \colon \mathcal{P}(S) \to \mathbb{R}$ be monotone submodular. Let $g \colon \mathcal{P}(V) \to \mathcal{P}(S)$ be a function satisfying $g(A\cup B) = g(A)\cup g(B)$ and $A \subset B \implies g(A) \subset g(B)$ for $A$, $B \subset V$. Then, $h \circ g \colon \mathcal{P}(V) \to \mathbb{R}$ is monotone submodular.

(Note that by $g(A)$ we do not mean $\{g(x)|x\in A\}$ but rather the image under $g$ of $A$ as an element $g(\{A\})$, but we omit the brackets. $A$ is a subset of $V$ but an element of $\mathcal{P(V)}$, and therefore $g$ sends it to a subset of $S$, element of $\mathcal{P}(S)$)
\label{teo2}
\end{lemma}
\cref{fig:diag} provides a scheme of the situation, in which we have used the explicit notation.
\begin{figure}[H]
    \centering
    \includegraphics[width=0.8\columnwidth]{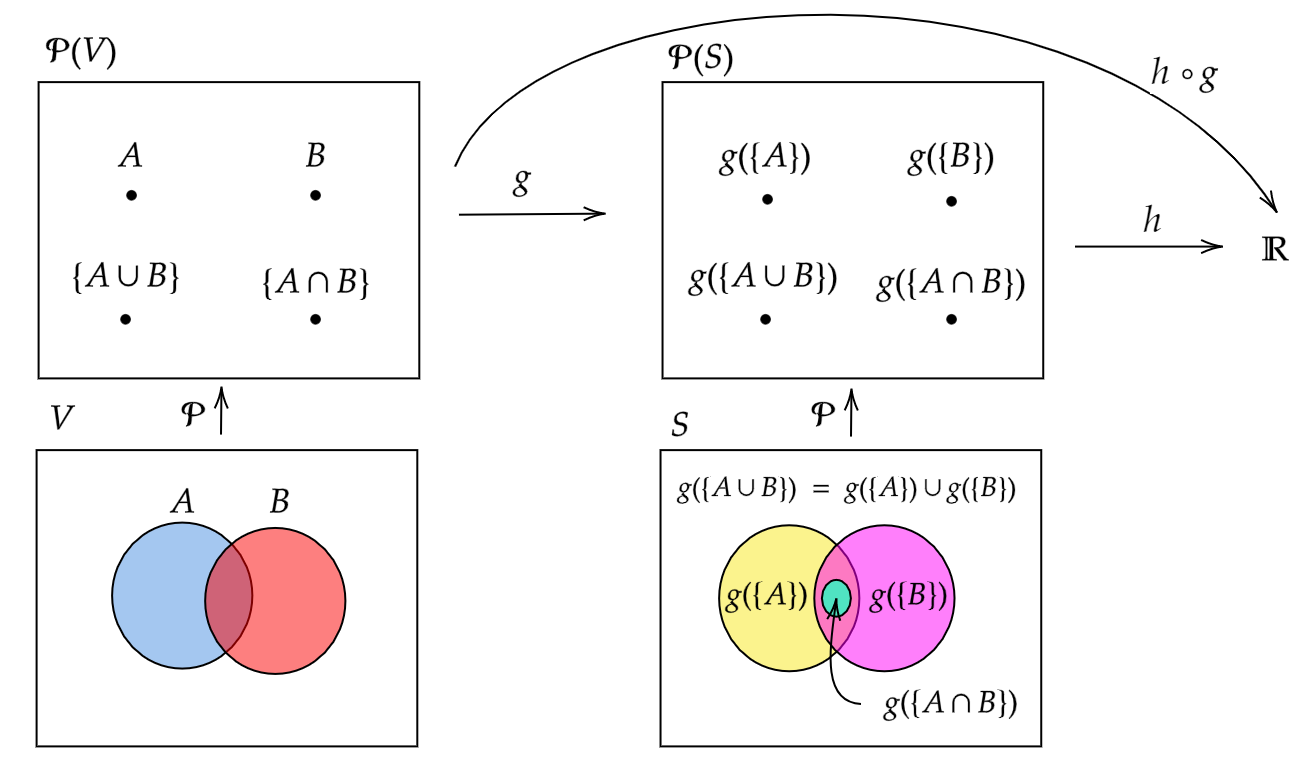}
    \caption{Scheme of \cref{teo2}. We want to see that given that $h$ is monotone submodular and $g$ satisfies certain conditions, the submodularity is conserved under the composition $h \circ g$}
    \label{fig:diag}
\end{figure}

\begin{proof}
The monotonicity of $h \circ g$ comes from the monotonicity of $h$ and $g$. For the submodularity, we want to see that for all $S, T \subset V$, 
\begin{align*}
(h \circ g) (S) + (h \circ g) (T) &\geq (h \circ g) (S\cup T) + (h \circ g) (S\cap T) \\
h(g(S)) + h(g(T)) &\geq h(g(S\cup T)) + h(g(S\cap T)) \\
h(g(S)) + h(g(T)) &\geq h(g(S) \cup g(T))) + h(g(S\cap T))
\end{align*}
Since $h$ is submodular and $g(S)$ and $g(T)$ are subsets of $S$, we know that 
\begin{align*}
h(g(S)) + h(g(T)) \geq h(g(S) \cup g(T))) + h(g(S)\cap g(T))
\end{align*}
But since $g$ is monotonous we have $g(S\cap T) \subset (g(S) \cap g(T))$ and since $h$ is monotonous $h(g(S\cap T)) \leq h(g(S) \cap g(T))$ and so we have $h(g(S)) + h(g(T)) \geq h(g(S) \cup g(T))) + h(g(S\cap T))$
\end{proof}

\paragraph{Proof of \cref{thm:submod}:}
We apply \cref{teo2} to the composition $h\circ D$, where $D$ is the detection set function $D \colon \mathcal{P(V)} \to \mathcal{P(S)}$ mapping $W$ to $D_W$, where $V$ is the set of vertices in the graph and $S$ the set of states of the stochastic process. By \cref{teo1} $h$ is submodular and it is also clearly monotonous by the same argument that we have shown that $f_\tau$ is monotonous. By definition of the detection set function $D$, $D_A \subset D_B$ if $A \subset B$ and $D_{A\cup B} = D_A \cup D_B$. Therefore, $f = h \circ D$ is submodular. 
$\square$.

\subsection{Proof of \cref{thm:submodfhat}}
\begin{theorem}
The sample approximation of $f_\tau$, $\hat{f}_\tau$ defined in \cref{f_hat_def}, is non-negative, monotone and submodular for all $N_R \in \mathbb{N}$
\end{theorem}

\begin{proof}
Non-negativity is true by definition, and monotonicity comes from the fact that if $A \subset B$, $min_{k\in A} L[i,k] \geq min_{k\in B} L[i,k]$ and so for a lesser or equal number of runs the minimum over $A$ will be less or equal than the minimum over $B$, so $f(A) \leq f(B)$. For submodularity, we want to prove that for $X\subset Y$ and $x \in V\setminus Y$
\begin{align}
\hat f (X\cup \{x\}) - \hat f (X) \geq \hat f (Y) - \hat f (Y\cup \{x\}) ~.
\end{align}
The left hand side is $$|i \in [N_R] \text{ s.t } \min_{k\in X\cup\{x\}}{L[i,k]} <= \tau| - |i \in [N_R] \text{ s.t } \min_{k\in X}{L[i,k]} <= \tau|$$ which equals $|i \in [N_R] \text{ s.t } (\min_{k\in X}{L[i,k]} > \tau) \cap (L[i,x] \leq \tau)|$. Similarly, the right hand side is $|i \in [N_R] \text{ s.t } (\min_{k\in Y}{L[i,k]} > \tau) \cap (L[i,x] \leq \tau)|$. As $\min_{k\in Y}{L[i,k]} > \tau \implies \min_{k\in X}{L[i,k]} > \tau$, there are at least as many elements in the set of the left-hand side than in the set of the right-hand side, proving the inequality.
\end{proof}
Taking limits in the submodularity inequality for $\hat f_\tau$ provides an alternative proof that $f_\tau$ is submodular.


\bibliographystyle{siamplain}
\bibliography{references}
\end{document}